\documentclass[12pt]{article}
\usepackage[margin=1in]{geometry}
\usepackage{setspace}
\usepackage{amsmath,amssymb,amsthm,enumitem,graphicx}

\newcommand{\etal}{\emph{et al.}}
\newcommand{\ecc}{\mathrm{ecc}}
\newcommand{\rad}{\mathrm{rad}}
\newcommand{\diam}{\mathrm{diam}}
\newcommand{\ds}{\mathrm{ds}}
\newcommand{\dsw}{\mathrm{ds}^w}
\newcommand{\lev}{\mathrm{lev}}
\newcommand{\ol}[1]{\overline{#1}}
\newcommand{\tw}[1]{\texttt{#1}}

\theoremstyle{definition}
\newtheorem{definition}{Definition}[section]
\newtheorem{note}[definition]{Note}
\newtheorem{example}[definition]{Example}
\theoremstyle{plain}
\newtheorem{theorem}[definition]{Theorem}
\newtheorem{proposition}[definition]{Proposition}
\newtheorem{lemma}[definition]{Lemma}
\newtheorem{corollary}[definition]{Corollary}

\title{Quasi-Isometric Graph Simplification}
\author{Roger Su \and Bakhadyr Khoussainov \and Simone Linz}
\date{\today}

\begin{document}


\maketitle

\begin{abstract}
Quasi-isometries are mappings on graphs, with distance-distortions parameterized by a multiplicative factor and an additive constant. The distance-distortions of quasi-isometries are in a general form that captures a wide range of distance-approximating graph simplifications. This paper introduces quasi-isometries into the field of graph simplifications, which is becoming increasingly important as large-scale graphs gain more and more prevalence. We discuss some general goals of graph simplification under the framework of quasi-isometries, and investigate several constructions of quasi-isometric graph simplifications, namely one based on maximal independent sets and one based on grouping vertices. For the latter construction, we prove that it preserves the centers and medians of trees.

\end{abstract}

\section{Introduction}\label{sec:intro}

In the last few decades, there has been a significant interest in the study of large-scale graphs that arise from modelling social networks, web graphs, biological networks, among others. Many of these graphs have millions or even billions of vertices or edges, so computing on these graphs is challenging task, as they render conventional algorithms inefficient or impractical.

The emergence of large-scale graphs urges us to revisit fundamental and graph concepts and algorithms, such as the number of vertices, the number of connected components, the distance between two given vertices, the sparseness or denseness, the degree distribution, the central vertices, and other global or local properties.

In order to tackle these questions for large-scale graphs, researchers devise graph simplification techniques, which are also called graph preprocessing or graph summarization~\cite{liu-18}. These techniques transform an input graph into a smaller graph or a concise data structure, such that computations on the input graph correspond to more efficient computations on the simplified structure. Different computational goals call for different graph preprocessing methods. For example, the efficiency of exact distance queries can be increased by adding so-called \emph{2-hop labels} to the vertices~\cite{cohen-03}, while the computation of a minimum cut can be improved by removing certain edges and resulting in a \emph{cut sparsifier}~\cite{benczur-15}.

Within the field of graph preprocessing, a prominent class of methods involve partitioning the vertices to aid visualization or computation. These methods come under the names of \emph{community detection}~\cite{fortunato-10}, \emph{graph partitioning}~\cite{bichot-13} and \emph{graph clustering}~\cite{kannan-04}. The primary goal of graph clustering is to discover dense subgraphs called ``clusters'' with few inter-cluster edges.
Preprocessing of graphs can also be viewed from the angle of data compression. \emph{Graph compression}~\cite{besta-18-survey,fan-22-contraction} is also based on grouping vertices, while storing additional information in order to recreate the original graph.

Among the various types of computations on graphs, the most fundamental is the distance query, which asks for the length of a shortest path between two vertices. The distance query is important because it is the basis of many other query types in fields such as \emph{transportation planning}~\cite{bast-16}, \emph{network design}~\cite{miller-13}, \emph{operational research}~\cite{slater-82} and \emph{graph databases}~\cite{graph-databases}. When the graph is large, direct computation of the distance becomes impractical, so there is a need for preprocessing methods that can efficiently answer distance queries with a certain degree of guaranteed accuracy.

There are may preprocessing methods to handle approximate distance queries. One of them is the construction of \emph{spanners}. On a connected graph $G$, an $(\alpha,beta)$-\emph{spanner} is a spanning subgraph $H$ with integer parameters $\alpha \geq 1$ and $\beta \geq 0$, such that for all vertices $v_1$ and $v_2$,
\begin{equation}\label{eq:spanner}
  d_H(v_1,v_2) \leq \alpha \cdot d_G(v_1,v_2) + \beta.
\end{equation}
Spanners were first introduced by Peleg and Sch\"affer~\cite{peleg-89}, who showed that it is NP-complete to compute the minimum $(\alpha,0)$-spanner, but for every integer $1\leq k\leq n$, there exists a $(4\log_k n + 1, 0)$-spanner with at most $kn$ edges and constructible in polynomial time. The field of spanners was subsequently expanded by many researchers on different algorithmic fronts, recently surveyed by Ahmed \etal~\cite{ahmed-20-spanner-survey}. 

The construction of spanners involves edge-deletions. On the other hand, one can also reduce a graph using edge-contraction. When $G$ is transformed into $H$ by contracting edges, there is a natural surjection $f$ from the vertices of $G$ to those of $H$, where adjacent vertices $v_1,v_2 \in G$ are mapped to a same vertex in $H$ if the edge $\{v_1,v_2\}$ is contracted.

Given real-valued constants $\alpha \geq 1$ and $\beta \geq 0$, Bernstein \etal~\cite{bernstein-19} studied the optimization problem of finding a minimal set of edges to contract, such that for all vertices $v_1$ and $v_2$ in $G$,
\begin{equation}\label{eq:contraction}
  \frac{1}{\alpha} \cdot d_G(v_1,v_2) - \beta = d_H\big( f(v_1), f(v_2) \big).
\end{equation}
Contraction can be viewed as assigning vertices into groups, which is related to graph clustering and graph compression as mentioned earlier~\cite{fan-22-contraction}.

Another distance-approximating graph preprocessing technique is \emph{distance oracles}, which were first invented by Thorup and Zwick~\cite{thorup-05}, and subsequently extended by many researchers, including the recent work by Charalampopoulos \etal~\cite{char-19}. While spanners and contractions produce smaller graphs, distance oracles are not themselves graphs, but are intricate data structures that can be constructed efficiently, and can return reasonably accurate distance queries.

The notion of distance-distortion is also related to \emph{metric embeddings}~\cite{metric-embed-abraham-11,metric-embed-ostrovskii-13}, a technique to develop approximation algorithms for graphs as well as general computational problems that involve metric spaces, such as graphs, computer vision and computational biology.
Metric embeddings typically map the original objects into the $\ell_p$-space, with the main goals being small dimension of the target space and small distance-distortions.
Embeddings are by nature injective and usually not surjective. This is different to graph simplification, which is surjective and not injective.

In this paper, we continue the lines of research on approximate distance-preservation, and introduce quasi-isometries to the active field of graph simplifications. A quasi-isometry is a mapping from one graph to another, and has parameters that satisfy a general inequality whose form captures both (\ref{eq:spanner}) and (\ref{eq:contraction}). Quasi-isometries will be formally introduced in Section~\ref{sec:q-iso}. Before then, we first summarize some general goals of graph simplifications in Section~\ref{sec:setup}.

\subsection{General Goals for Graph Simplification}\label{sec:setup}

The aim of quasi-isometric graph simplification is to find a quasi-isometry of a large-scale graph $G$ into some smaller but non-trivial graph $H$, such that the quasi-isometry has small constants, and at the same time the graph $H$ retains important properties of the original graph $G$.

Let $Q$ be an abstract property of graphs, which can be a \emph{local} predicate on the vertices (such as ``being a central vertex'' or ``having the maximum degree'') or a \emph{global} property of the graph (such as ``being a tree'' or ``being chordal'').

Let $\mathcal{K}$ be a class of graphs. Given $G\in\mathcal{K}$ and a property $Q$. The aim is to construct a smaller graph $H$ with a quasi-isometry $f : G \rightarrow H$ such that the following properties are satisfied.

\begin{enumerate}
  \item \textbf{Small quasi-isometry constants}: This property controls the distortion between the distance-functions in $G$ and $H$. This also avoids collapsing $G$ into the trivial singleton-graph, as collapsing a graph into the singleton-graph requires large constants.
  \item \textbf{Compression}: By requiring the number of vertices in the simplified graph $H$ to be a fraction of the number of vertices in $G$, this ensures that $H$ is a meaningful size-reducing simplification.
  \item \textbf{Preservation}: The property $Q$ should be well-behaved with respect to $f$. This is free to a reasonable interpretation. For instance, one can demand that for all $x\in G$, if $x$ satisfies $Q$, then the vertex $f(x)$ should also satisfy $Q$.
  \item \textbf{Efficiency}: Building $f$ and $H$ should be efficient on the size of $G$. From an algorithmic view point, this property is natural and is directly linked with issues related to the preprocessing of large-scale graphs.
  \item \textbf{Retention}: $H$ should be in the same class $\mathcal{K}$ as $G$. In other words, the quasi-isometry $f$ should retain the key algebraic properties of $G$.
\end{enumerate}
These are the general goals of quasi-isometric graph simplifications. In later parts of the paper we will refer to them.

\subsection{Our Contributions}\label{sec:outline}

Below we list our contributions to the area on large-scale graphs, with connections to the list in Section~\ref{sec:setup}.
\begin{enumerate}
  \item We propose a general formal framework to study large-scale graphs based on quasi-isometries. We provide several simple constructions that quasi-isometrically map large graphs to smaller graphs.
  \item In order to build quasi-isometries of graphs, we introduce the notion of \emph{partition-graphs} in Section \ref{sec:partition-graphs}. These are simplified graphs built from any given graph $G$ by grouping vertices. We show that partition-graphs are quasi-isometric to the original graph, with the quasi-isometry constants depending on the diameters of the super-vertices, and the compression property depending on the cardinalities of the super-vertices.
  \item We investigate the question if the vertices in the center of a given graph are preserved under quasi-isometries. Among the countless different notions of graph centrality, we focus on the two most basic: the \emph{center} and the \emph{median}, both of which are defined in terms of the distance. In order to capture the effect of graph simplifications on the center, Section~\ref{sec:center-shift} introduces the concept called the \emph{center-shift}. Given a quasi-isometric graph simplification $\varphi:G\rightarrow H$, with $C_G$ and $C_H$ being the respective centers of $G$ and $H$, the center-shift measures the distance in $G$ between $C_G$ and $\varphi^{-1}(C_H)$.
  \item It turns out that a quasi-isometry alone is not strong enough to bound the center-shift. As shown by Theorem~\ref{thm:center-shift-2side}, the center-shift of a mapping $\varphi:G\rightarrow H$ is bounded above by a function involving the radius of $H$, given $G$ satisfies a special property.
  \item  We then focus on trees.
  Trees already provide interesting counter-examples (such as Figure~\ref{fig:tree-comb}), which suggests that even on trees, quasi-isometric simplifications needs to be constructed with care, depending on the query.
  Furthermore, although trees are simple objects, they are nevertheless used in many fields such as mathematical phylogenetics~\cite{semple-03} and optimization~\cite{wu-chao-04}.
  Section~\ref{sec:outward-contraction} shows that the method of \emph{outward-contraction} produces partition-trees with center-shift zero, which means that outward-contraction preserves the centers of trees. 
  \item Finally, Section~\ref{sec:vertex-weights} shows that to preserve the median of trees, we need to store extra numerical information. Without this numerical information, there are cases where outward-contraction does not preserve the medians of trees. However, if we store the cardinality of each vertex-subset in the partition, and handle the graph as a vertex-weighted graph, then we can locate the median of the original tree from the partition.
\end{enumerate}
In terms of our problem set-up in Section~\ref{sec:setup}, our contributions 1--2 address issues related to small quasi-isometry constants and compression, while 3--6 focus on the preservation of the center (as the property $Q$) under specific quasi-isometric simplification of trees.

\section{Preliminaries}\label{sec:prelim}

In this paper, all graphs are assumed to be undirected, finite, and without loops or parallel edges. In formal terms, a \emph{graph} $G$ is a pair $\langle V(G), E(G) \rangle$, where $V(G)$ is a finite set of \emph{vertices}, and $E(G)$ is a set of \emph{edges} (which are 2-element subsets of $V(G)$).
Two vertices $v_1$ and $v_2$ are \emph{adjacent}, denoted $v_1\sim v_2$, when $\{v_1,v_2\} \in E(G)$.
We sometimes write $G$ to mean the vertex-set $V(G)$ when no ambiguity arises, and $|G|$ denotes the number of vertices in $G$.

In a graph, a \emph{path} $\mathbf{v}$ is a sequence of vertices $v_1,v_2,\ldots v_{\ell}$ such that $v_i \sim v_{i+1}$ for all $1\leq i<\ell$. A \emph{simple path} is a path with all vertices distinct. Every path can be reduced to a simple path with the same endpoints. The \emph{length} of a (simple) path is the number of edges. Two vertices are \emph{connected} when they are endpoints of some path. A graph is \emph{connected} when every pair of vertices is connected. In this paper, all graphs are assumed to be connected.

The \emph{path-graph} on $n$ vertices, denoted $P_n$, is the graph on vertices $v_1,v_2,\ldots v_n$ such that $v_i \sim v_{i+1}$ for all $1\leq i<n$.


The \emph{distance} between two vertices $v_1$ and $v_2$, denoted $d(v_1,v_2)$, is the length of a shortest path between $v_1$ and $v_2$.
For two vertex-sets $S_1,S_2 \subseteq V(G)$, the distance $d(S_1,S_2)$ is defined to be $\min\{ d(v_1,v_2) \mid v_1\in S_1, v_2\in S_2 \}$,
while the distance between a vertex $v$ and a vertex-set $S$ is $d(v,S) = \min\{ d(v, x) | x\in S \}$.


\begin{definition}[Center]\label{defn:center}
The \emph{eccentricity} of a vertex $v$ is the maximum distance from $v$ to any other vertex: $\ecc(v) = \max\{ d(v,x) \mid x\in G \}$.
The \emph{eccentricity-witnesses (ecc-wits)} of a vertex $v$ are the vertices $x$ such that $\ecc(v) = d(v,x)$.

The \emph{center} of a graph $G$ (denoted $C_G$) is the set of vertices with the minimum eccentricity.
\end{definition}

\begin{proposition}\label{prop:ecc-diff}
For vertices $v_1,v_2$ in a graph, $|\ecc(v_1) - \ecc(v_2)| \leq d(v_1,v_2)$.
\end{proposition}

\begin{proof}
Without loss of generality, assume $\ecc(v_1) \geq \ecc(v_2)$. Also let $\ecc(v_1) = d(v_1, v_1')$ and $\ecc(v_2) = d(v_2, v_2')$. Now,
\begin{align*}
   \ecc(v_1) - \ecc(v_2)
&= d(v_1,v_1') - d(v_2,v_2') \\
&\leq d(v_1,v_1') - d(v_2,v_1') &&\text{($v_2'$ is an ecc-wit of $v_2$)}\\
&\leq d(v_1,v_2) &&\text{(triangle inequality).}
\end{align*}
\end{proof}

\begin{note}\label{note:leaf-removal}
It is well known that the center of a tree consists of a single vertex or two adjacent vertices, and that the center of a tree can be located by a \emph{leaf-removal} algorithm~\cite{goldman-71}. At the start, leaf-removal removes all the leaves of the input tree $T$, and results in a smaller tree $T_1$. Next, leaf-removal removes all the leaves of $T_1$ and results in $T_2$. This process repeats until only a single vertex or two adjacent vertices remain. Then the final remaining vertices are the center of $T$. Later parts of this paper will refer to this algorithm.
\end{note}

The \emph{radius} of $G$ is the minimum eccentricity: $\rad(G) = \min\{ \ecc(x) \mid x\in G \}$.
The \emph{diameter} of $G$ is the maximum eccentricity: $ \diam(G) = \max\{ \ecc(x) \mid x\in G \}$.
A \emph{diameter-path} is a path whose length equals the diameter.

\begin{definition}[Distance-sum and median]
The \emph{distance-sum} of a vertex $v$ is defined to be \[ \ds(v) = \sum_{x\in G} d(v,x) .\]
The \emph{median} of a graph is the set of vertices with the minimum distance-sum.
\end{definition}

\subsection{Quasi-Isometries}\label{sec:q-iso}

Our aim is to devise a general framework for distance-preserving graph simplifications, and we base our framework on the notions of \emph{quasi-isometries} and \emph{large-scale geometries} introduced by Gromov~\cite{gromov-81} in algebraic geometry.
Quasi-isometries can be viewed as bi-Lipschitz maps on metric spaces with a finite additive distortion, and they lead to the concept of large-scale geometries, which turned out to be crucial in the study of infinite algebraic objects that are finitely generated (such as groups and their Cayley graphs).
Later, quasi-isometries and large-scale geometries were applied to infinite trees by Kr\"on and M\"oller~\cite{kroen-08} as well as to infinite strings by Khoussainov and Takisaka~\cite{khou-17}.
We first recall the definition of metric spaces, on which quasi-isometries are defined.

\begin{definition}[Metric space]\label{def:metric-space}
A \emph{metric space} $\langle M, d \rangle$ consists of a set $M$ and a function $d$ that maps every two elements in $M$ to a non-negative real numbers. The function $d$ is called the \emph{distance} or the \emph{metric}, and is required to satisfy the following axioms:
\begin{itemize}
  \item (M1) $\forall x,y\in M: x=y$ if and only if $d(x,y)=0$.
  \item (M2) $\forall x,y\in M: d(x,y) = d(y,x)$.
  \item (M3) $\forall x,y,z\in M: d(x,z) \leq d(x,y) + d(y,z)$.
\end{itemize}
\end{definition}

Every connected graph is a metric space with the shortest-path distance.

For metric spaces $\langle M_1, d_1 \rangle$ and $\langle M_2, d_2 \rangle$, a mapping $f:M_1\rightarrow M_2$ is called an \emph{isometry} when
\[ \forall x,y\in M_1 : d_1(x,y) = d_2\big( f(x), f(y) \big) .\]
A quasi-isometry is a more general version of an isometry, as stated in the following definition, first attributed to Gromov~\cite{gromov-81}.

\begin{definition}[Quasi-isometry]\label{def:q-iso}
Let $\langle M_1, d_1 \rangle$ and $\langle M_2, d_2 \rangle$ be metric spaces, and let $A,B,C$ be non-negative integers with $A\geq 1$.
Then a function $f:M_1\rightarrow M_2$ is called an \emph{$(A,B,C)$-quasi-isometry} if the following two properties are satisfied:
\begin{itemize}
  \item (Q1) $\forall x,y\in M_1$:
  \[ \frac{1}{A} \cdot d_1(x,y) - B \leq
     d_2 \big( f(x), f(y) \big) \leq
     A \cdot d_1(x,y) + B .\]
  \item (Q2) $\forall y\in M_2 : \exists x\in M_1 : d_2 \big( y, f(x) \big) \leq C$. 
\end{itemize}	
Property (Q1) is called the \emph{quasi-isometric inequality}, which is the general form that captures the inequalities of spanners (\ref{eq:spanner}) and distance-preserving graph contractions (\ref{eq:contraction}) in Section~\ref{sec:intro}.

Property (Q2) is called the \emph{density property}, which is more of a technical device to ensure that quasi-isometries form a symmetric relation. If $f$ is not surjective and $y$ is not mapped to by $f$, then $y$ is not too far away from an element that is mapped to by $f$. As Proposition~\ref{prop:q-iso} will show, the density property is trivial when $f$ is surjective. In the context of graph simplification, the mapping is always surjective, so the density property is not of central importance in this paper.

Among the \emph{quasi-isometry constants} $A$, $B$, $C$, the constant $A$ is called the \emph{stretch factor}, and $B$ the \emph{additive distortion}.
\end{definition}

\begin{proposition}\label{prop:q-iso}
A mapping $f$ on metric spaces satisfies the following:
\begin{enumerate}
  \item $f$ is a $(1,0,0)$-quasi-isometry if and only if $f$ is an isometry.
  \item $f$ is an $(A,B,0)$-quasi-isometry if and only if $f$ is surjective.
  \item $f$ is an $(A,0,C)$-quasi-isometry implies $f$ is injective.
\end{enumerate}
\end{proposition}

\begin{proof}
Firstly, $(1,0,0)$-quasi-isometry is an isometry simply by definition. For the second fact, note that $C=0$ makes property (Q2) become \[ \forall y\in M_2 ~:~ \exists x\in M_1 ~:~ d_2 \big( y, f(x) \big) = 0 .\]
However, by axiom (M1), $d_2 ( y, f(x) ) = 0$ if and only if $y = f(x)$, which is precisely surjectivity.

Finally, to show that $B=0$ implies injectivity, let $x,y\in M_1$ with $x\neq y$. Then axiom (M1) in $M_1$ ensures $0 < d_1( x, y )$. Since $A = \alpha \geq 1$ and $B = 0$, the inequality $0 < 1/\alpha \cdot d_1( x, y )$ also holds. Then (Q1) leads to \[ 0 < \frac{1}{\alpha} \cdot d_1( x, y ) < d_2 \big( f(x), f(y) \big) ,\] which means that $f(x) \neq f(y)$ due to axiom (M1) in $M_2$. Therefore $f$ is injective.
Note that the converse of the injectivity statement does not always hold.
\end{proof}

Quasi-isometries form an equivalence relation on metric spaces. Let two metric spaces $M_1$ and $M_2$ be related when there \emph{exist} constants $A$, $B$ and $C$ such that $M_1$ is mapped to $M_2$ with an $(A,B,C)$-quasi-isometry.
Then with long but routine reasonings, one can show that this relation satisfies reflexivity, symmetry and transitivity. As a consequence, all finite metric spaces are quasi-isometric to each other, and form one equivalence class under quasi-isometries.
Namely, every finite metric space $M$ is quasi-isometric to the singleton space via the function $f$ that trivially maps every element in $M$ to the singleton space's only element. The constants of this mapping are $A=1$, $B=\diam(M)$ and $C=0$, since $f(x) = f(y)$ for all $x,y\in M$, and hence $d'\big( f(x), f(y) \big) = 0$.
Then we need constants $A$ and $B$ such that \[ \frac{1}{A} \cdot d(x,y) - B \leq 0 \leq A \cdot d(x,y) + B ,\] which can always be satisfied when $A=1$ and $B=\diam(M)$. As for the final constant, $C=0$ because $f$ is surjective (Proposition~\ref{prop:q-iso}).

The term \emph{large-scale geometries} denotes the equivalence classes under quasi-isometries. Under the original context by Gromov, the large-scale geometry of all finite graphs are trivial. Nevertheless, while every finite graph is equivalent to the singleton graph, the quasi-isometric constants may be arbitrarily large. Thus this paper is a refinement of Gromov's idea, as it requires the quasi-isometric constants to be \emph{small}.

\section{Quasi-Isometries via Independent Sets}
\label{sec:q-iso-ind-set}

Quasi-isometric versions of a given graph can be constructed in different ways. This section introduces a construction via maximal independent sets. Another construction based on grouping vertices will be presented later Section~\ref{sec:partition-graphs}.

In a graph $G$, a subset of vertices $S \subset G$ is called an \emph{independent set} when every two vertices in $S$ are not adjacent. Furthermore, $S$ is a \emph{maximal} independent set (MIS) if $S\cup\{v\}$ is no longer an independent set, for every $v\in G\setminus S$.

\begin{definition}[MIS-derived graph]
Let $G$ be a connected graph, and let $S \subset G$ be a maximal independent set in $G$. Then the \emph{MIS-derived graph} $S$ is defined with $S$ as its vertex-set, and $x,y\in S$ are adjacent in $S$ when $d_G(x,y) \leq 3$.
\end{definition}

Before investigating any distance-related properties in the MIS-derived graph $S$, we first need to establish that $S$ is connected, so that the distance-function on $S$ is well defined.

\begin{lemma}
The MIS-derived graph $S$ is connected.
\end{lemma}

\begin{proof}
Let $x,y\in S\subset G$. The original graph $G$ is assumed connected, so there exists at least one simple $G$-path between $x$ and $y$. Now, the goal is to show the existence of an $S$-path between $x$ and $y$.

Let $x = v_0$ and $y = v_k$, and let $\mathbf{v} = v_0, v_1, \ldots v_k$ be a $G$-path between $x$ and $y$. We say that $\mathbf{v}$ corresponds to an $S$-path when there is an $S$-path whose vertices are all on $\mathbf{v}$.

Now, if $\mathbf{v}$ corresponds to an $S$-path, then this $S$-path shows that $x$ and $y$ are connected. If $\mathbf{v}$ does not correspond to an $S$-path, then it contains (potentially multiple) pairs of vertices $v_i$ and $v_j$ that satisfy:
\begin{itemize}
  \item $0\leq i$ and $i\leq j-4$ and $j\leq k$.
  \item $v_i, v_j \in S$.
  \item for all $i<p<j: v_p\notin S$.
\end{itemize}
Figure~\ref{fig:mis-special-case} illustrates the general form of these pairs.
Here, $v_i$ and $v_j$ are not adjacent in $S$ because $d_G (v_i, v_j) > 3$. Moreover, even if they are connected, any intermediate vertex cannot be on $\mathbf{v}$ because every vertex $v_p$ between them is not in $S$.

\begin{figure}[t]
  \centering
  \includegraphics[scale=0.5]{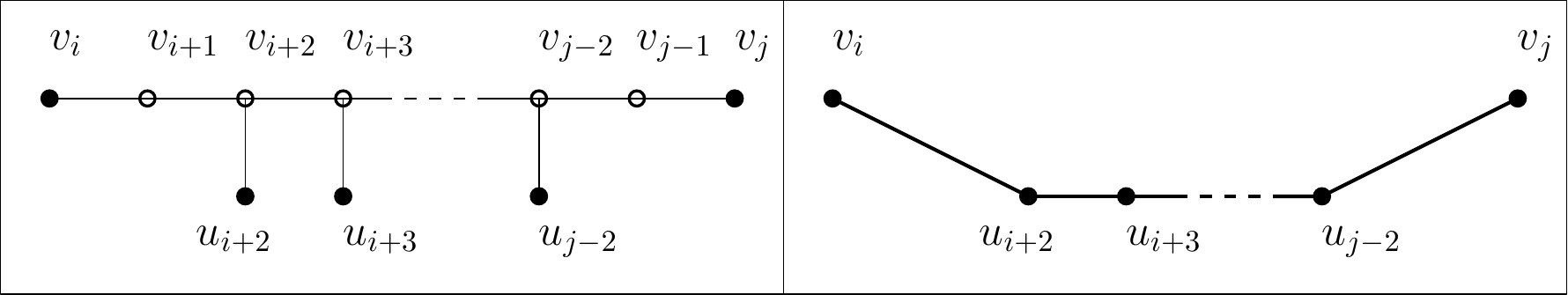}
  \caption{Left: Part of a graph $G$, shaded vertices are in the MIS $V(S)$, the $G$-path between $v_i$ and $v_j$ does not correspond to an $S$-path. Right: Part of the MIS-derived graph, the $S$-path not on the original $G$-path.}
  \label{fig:mis-special-case}
\end{figure}

Consider every $v_p$ with $i+2 \leq p \leq j-2$. Given $v_p\notin S$ and that $S$ is an MIS, $v_p$ must have a neighbor in $S$. The two neighbors of $v_p$ on $\mathbf{v}$ are not in $S$, so $v_p$ must have a neighbor $u_p$ that is in $S$ and not on $\mathbf{v}$. Then, \[ v_i \sim u_{i+2} \sim u_{i+3} \sim \cdots \sim u_{j-2} \sim v_j \text{~in~} S,\] which shows that $v_i$ and $v_j$ are connected via an $S$-path that does not correspond to $\mathbf{v}$. Applying this construction to every pair of such $v_i$ and $v_j$, we obtain an $S$-path that does not completely correspond to $\mathbf{v}$ but still connects $x$ and $y$.
\end{proof}

\begin{definition}[Mapping to MIS-derived graph]\label{defn:mis-f}
Let $f:G\rightarrow S$ such that for all $v\in V(G)$:
\begin{enumerate}
  \item if $v\in V(S)$, then $f(v)$ is simply $v$;
  \item otherwise, $f(v)$ is set to be any $w$ such that $w\in V(S)$ and $w \sim_G v$.
\end{enumerate}
Since $V(S)$ is a maximal independent set, every vertex not in $V(S)$ must have a neighbor in $V(S)$, so the second case is well defined.
\end{definition}

\begin{theorem}\label{thm:ind-set-ineq}
For all $x,y\in V(G)$,
\[      \left\lfloor \frac{d_G\big( x, y \big)}{3} \right\rfloor + 1
   \leq d_S\big( f(x), f(y) \big)
   \leq d_G\big( x, y \big) .
\]
\end{theorem}

\begin{proof}
Each of $x$ and $y$ may or may not belong to $S$, but by Definition~\ref{defn:mis-f}, $d_G \big( x, f(x) \big) \leq 1$ and $d_G \big( y, f(y) \big) \leq 1$. Then with the triangle inequality, we obtain
\begin{equation}\label{eq:mis-xy-f}
  d_G \big( x, y \big) - 2 \leq d_G \big( f(x), f(y) \big) \leq d_G \big( x, y \big) +2 .
\end{equation}
Given a shortest $G$-path between $f(x)$ and $f(y)$, the longest possible $S$-path has the form in Figure~\ref{fig:mis-special-case}, so
\begin{equation}\label{eq:mis-upper}
 d_S \big( f(x), f(y) \big) \leq d_G \big( f(x), f(y) \big) - 2 .
\end{equation}
As for the shortest possible $S$-path, this is obtained when the $G$-path between $f(x)$ and $f(y)$ is decomposed into as many three-edge segments as possible. Hence 
\begin{equation}\label{eq:mis-lower}
  \left\lfloor \frac{d_G \big( f(x), f(y) \big)}{3} \right\rfloor + 1 \leq d_S \big( f(x), f(y) \big) .
\end{equation}
Combining (\ref{eq:mis-xy-f}), (\ref{eq:mis-upper}) and (\ref{eq:mis-lower}) with some routine rearranging yields the theorem's statement.
\end{proof}

\begin{corollary}\label{cor:ind-set-q-iso}
The mapping $f$ is a $(3,1,0)$-quasi-isometry.
\end{corollary}

\begin{proof}
The inequality in Theorem~\ref{thm:ind-set-ineq} can be relaxed into
\[ \frac{1}{3} \cdot d_G\big( x, y \big) - 1
   \leq d_S\big( f(x), f(y) \big) \leq
   3 \cdot d_G\big( x, y \big) + 1 , \]
so the multiplicative stretch is 3 while the additive distortion is 1.
Finally, the third constant is zero because $f$ is surjective (Proposition~\ref{prop:q-iso}).
\end{proof}

Since quasi-isometries form an equivalence relation, Corollary~\ref{cor:ind-set-q-iso} implies that all the graphs derived from the maximal independent sets form one equivalence class.

A maximal independent set can have as many as $n-1$ vertices, which is witnessed in the case of the star-graph. Moreover, a star-graph is acyclic, but the MIS made up of the non-center vertices induces a complete graph. These two observations show that the construction of MIS-derived graphs does not always satisfy the \emph{compression} and \emph{retention} properties listed in Section~\ref{sec:setup}, and therefore is not a good simplification method.

\section{Quasi-Isometries via Partition-Graphs}\label{sec:partition-graphs}

This section introduces another construction of quasi-isometric graph simplification. This method, based on grouping vertices, will be the focus of the rest of this paper. 

\begin{definition}\label{def:partition}
A \emph{partition} of a graph $G$ is a partition of $V(G)$ into subsets that induce connected subgraphs. These subsets are called \emph{super-vertices}.
\end{definition}

Note that the word ``partition'' here is slightly different from the set-theoretic use of the same word, as we additionally require the super-vertices to induce connected subgraphs.

\begin{definition}[Partition-graph]\label{def:partition-graph}
Given a partition $\mathcal P$ on $G$, the \emph{partition-graph} $G(\mathcal{P})$ is defined as follows.
\begin{itemize}
  \item The vertices of $G(\mathcal{P})$ are the connected subsets in the partition $\mathcal{P}$.
  \item Two super-vertices $P_i$ and $P_j$ are adjacent via a \emph{super-edge} in $G(\mathcal{P})$ if and only if there exist $x\in P_i$ and $y\in P_j$ such that $x\sim y$ in $G$.
\end{itemize}

Sometimes $\ol{G}$ is written instead of $G(\mathcal P)$ when $\mathcal P$ is clear from the context. For any $v\in G$, $\ol{v}$ denotes the super-vertex in $\ol{G}$ that contains $v$. This leads to the natural mapping $\varphi:G \rightarrow \ol{G}$ with $\varphi(v)=\ol{v}$.
\end{definition}

\begin{proposition}\label{prop:pgrf-path-leq}
Let $\mathcal{P}$ be any partition on a graph $G$, and let $H$ denote $G(P)$. Then for every pair $x,y\in G$ satisfies $d_H(\ol{x}, \ol{y}) \leq d_G(x,y)$.
\end{proposition}

\begin{proof}
Let $d_G(x,y) = k$, $x = v_0$ and $y = v_k$. Also let $\mathbf{v} = v_0, v_1, \dots v_k$ be a shortest simple path in $G$. Consider the path $\ol{\mathbf{v}} = \ol{v}_0, \ol{v}_1, \dots \ol{v}_k$ in $H$. The path $\ol{\mathbf{v}}$ in $H$ may not be a simple path; it has possible repeats because the mapping from $v$ to $\ol{v}$ is not injective in general. Hence, when $\ol{\mathbf{v}}$ is shortened to a simple path in $H$, its length $\leq k$.
\end{proof}

A \emph{simple cycle} $v_1,\ldots v_k, v_1$ is made up of a simple path $v_1,v_2,\ldots v_k$ and an edge $v_k \sim v_1$. Now the reasoning on paths in Proposition~\ref{prop:pgrf-path-leq} can be straightforwardly applied to cycles to obtain Corollary~\ref{cor:pgrf-cycle-leq}.

\begin{corollary}\label{cor:pgrf-cycle-leq}
The longest simple cycle in $G$ has length $k$ if and only if the longest simple cycle in $G(\mathcal P)$ has length $\leq k$.
\end{corollary}

As trees has no simple cycle, Corollary~\ref{cor:pgrf-cycle-leq} leads to to Corollary~\ref{cor:pgrf-tree-retention}, which shows that the construction of partition-graphs satisfies the \emph{retention} property (Section~\ref{sec:setup}) on trees.

\begin{corollary}\label{cor:pgrf-tree-retention}
If $G$ is a tree, then $G(\mathcal P)$ is also a tree for any partition $\mathcal P$.
\end{corollary}

Next, in order for $\varphi$ to be a quasi-isometry, the partition needs to have an upper bound on the diameters of all its super-vertices, as stated in Definition~\ref{defn:sharp-partition}. This is akin to chopping the original graph into bits that are small and ``sharp''.

\begin{definition}[Sharp partition]\label{defn:sharp-partition}
Given a natural number $c$, a partition is called a \emph{$c$-sharp partition} when every super-vertex $\ol{v}$ satisfies $\diam(\ol{v}) \leq c$. 
\end{definition}

\begin{theorem}\label{thm:sharp-partition-q-iso}
If $\mathcal{P}$ is a $c$-sharp partition on $G$, then the natural mapping $\varphi$ from $G$ to $H=G(\mathcal{P})$ is a $(c+1,1,0)$-quasi-isometry.
\end{theorem}

\begin{proof}
Let $x,y\in G$. Any shortest $H$-path between $\ol{x}$ and $\ol{y}$ corresponds to at most $(c+1) \cdot d_H(\ol{x},\ol{y}) + c$ many $G$-edges between $x$ and $y$, so $d_G(x,y) \leq (c+1) \cdot d_H(\ol{x},\ol{y}) + c$. This can be rearranged into
\begin{equation}\label{eq:pgrf-q-iso-lower}
       \frac{1}{c+1} \cdot d_G(x,y) - 1
  ~\leq~ \frac{1}{c+1} \cdot d_G(x,y) - \frac{c}{c+1} 
  ~\leq~ d_H(\ol{x},\ol{y}) .
\end{equation}
On the other hand, Proposition~\ref{prop:pgrf-path-leq} established $d_H(\ol{x}, \ol{y}) \leq d_G(x,y)$, which can then be relaxed to
\begin{equation}\label{eq:pgrf-q-iso-upper}
  d_H(\ol{x}, \ol{y}) ~\leq~ d_G(x,y) ~\leq~ (c+1)\cdot d_G(x,y) + 1 .
\end{equation}
Combining (\ref{eq:pgrf-q-iso-lower}) and (\ref{eq:pgrf-q-iso-upper}), we can conclude that the stretch factor is $c+1$ and the additive distortion is $1$. Meanwhile, the third constant is zero because $\varphi$ is surjective. Therefore, $\varphi$ is a $(c+1,1,0)$-quasi-isometry.
\end{proof}

The sharpness of a partition ensures small quasi-isometry constants, the first goal in Section~\ref{sec:setup}. However, when the sharpness is zero, the partition-graph is exactly identical to the original graph, and does not achieve any meaningful simplification. In order to satisfy the \emph{compression} property, we need the concept of a partition's coarseness (Definition~\ref{defn:coarse-partition}) in addition to the sharpness.

\begin{definition}[Coarse partition]\label{defn:coarse-partition}
For a natural number $b\geq 0$, a partition $\mathcal P$ of a graph $G$ is called a \emph{$b$-coarse partition} when every super-vertex $\ol{v}$ satisfies $\diam(\ol{v}) \geq b$. 
\end{definition}

If a super-vertex's diameter is at least $b$, then it must contain at least $b+1$ vertices. Hence, if $\mathcal P$ is a $b$-coarse partition, then \[ |G(\mathcal P)| \leq \frac{1}{b+1} \cdot |G| .\]
In conclusion, a small sharpness value ensures small quasi-isometry constants, and implies good distance-approximation. On the other hand, a large coarseness value ensures sufficient \emph{compression}. These respectively correspond to the first two goals listed in Section~\ref{sec:setup}, so a good partition must achieve a balance between these two antithetical parameters.

\begin{example}
There are different ways of constructing partitions. Here we present a simple construction.

On any input graph, we let every vertex be \emph{unassigned} initially. Then, choose any unassigned vertex $v$, and assign $v$ and its unassigned neighbors to a new super-vertex. Repeat this process until no unassigned vertex remains. 

Now, since the diameter of every collapsed neighborhood is at most two, the resulting partition is 2-sharp.
However, this can produce super-vertices that contain only one vertex, and potentially result in a 0-coarse partition, which does not satisfy the \emph{compression} property from Section~\ref{sec:setup}.

To remedy this, we can make a modification. Define an unassigned vertex to be \emph{completely free} when all of its neighbors are unassigned. Then the modified method runs as follows:
\begin{enumerate}
  \item While there exists some completely free vertex, choose any completely free vertex $v$, and assign $v$ and its unassigned neighbors to a new super-vertex.
  \item Then we reach a stage where all of the remaining unassigned vertices are not completely free. For each unassigned vertex $w$, it must have an assigned neighbor. Hence, we choose any assigned neighbor $u$, and place $w$ into the super-vertex containing $u$.
\end{enumerate}
The resulting partition is 4-sharp and 2-coarse. This means that the \emph{compression} property is satisfied, while the quasi-isometry constants are still small. Therefore, this modified method achieves the goals better.
\end{example}

\section{Center-Shift}\label{sec:center-shift}

This section studies the \emph{preservation} property of the graph center (Definition~\ref{defn:center}) under partition-graphs. Suppose $G$ is a large-scale graph whose center $C_G$ is impractical to compute. Then one wants to simplify $G$ to a smaller graph $H$ with a surjective mapping $\varphi : G \rightarrow H$, locate the center of $H$ (denoted $C_H$), and then infer the center $C_G$ using $\varphi$ and $C_H$. A natural way to infer $C_G$ is to use the reverse image: $\varphi^{-1}(C_H) = \{ v \in G \mid \varphi(v) \in C_H \}$, which is well defined when $\varphi$ is surjective.

The set of vertices $\varphi^{-1}(C_H)$ does not necessarily equal the original center $C_G$. Therefore, a metric is needed to measure how far apart $\varphi^{-1}(C_H)$ and $C_G$ are. For this, we introduce the \emph{center-shift} in Definition~\ref{defn:center-shift}. Since $H$ is the simplified and ``coarser'' graph, defining the center-shift in terms of the distance in $G$ seems more accurate and reasonable.

\begin{definition}[Center-shift]\label{defn:center-shift}
The \emph{center-shift} of a surjective mapping $\varphi:G\rightarrow H$ is defined to be $d_G \big( C_G , \varphi^{-1}(C_H) \big)$.
\end{definition}

Before investigating any possible relationship between quasi-isometry and the center-shift, we first present Lemma~\ref{lma:q-iso-ecc}, which shows that the form of the quasi-isometric inequality applies not only to the distance-function but also to the eccentricity.

\begin{lemma}\label{lma:q-iso-ecc}
Let $\varphi:G\rightarrow H$ be an $(A,B,C)$-quasi-isometry. Then for all  $v\in G$:
\[ \frac{1}{A} \cdot \ecc_G(v) - B
   \leq \ecc_H \big( \varphi(v) \big)
   \leq A \cdot \ecc_G(v) + B. \]
\end{lemma}

\begin{proof}
First, let $v'\in G$ be an ecc-wit of $v$. Then as $\ecc_G(v) = d_G(v,v')$,
\begin{equation}\label{eq:q-iso-ecc-1}
   \frac{1}{A} \cdot \ecc_G(v) - B = \frac{1}{A} \cdot d_G(v,v') - B
   \leq d_H \big( \varphi(v), \varphi(v') \big) \leq \ecc_H \big( \varphi(v) \big) .
\end{equation}
On the other hand, let $u\in G$ such that $\varphi(u)$ is an ecc-wit of $\varphi(v)$. Then
\begin{equation}\label{eq:q-iso-ecc-2}
   \ecc_H \big( \varphi(v) \big)= d_H \big( \varphi(v), \varphi(u) \big) \leq A \cdot d_G (v,u) - B \leq A \cdot \ecc_G (v) - B .
\end{equation}
And combining (\ref{eq:q-iso-ecc-1}) and (\ref{eq:q-iso-ecc-2}) completes the proof.
\end{proof}

It turns out that Definition~\ref{def:uni-ecc} is required in order to derive an upper bound on the center-shift, starting from just a quasi-isometry.

\begin{definition}[Uni-ecc]\label{def:uni-ecc}
Let $G$ be a graph with center $C_G$. Then $G$ is said to have the \emph{uniform eccentricity (uni-ecc)} property when $d_G(C_G,v) = \ecc(v) - \rad(G)$ for all $v\in G$. 
\end{definition}

The uni-ecc property is in fact a strong property that is not satisfied by most graphs. This will be discussed later in Section~\ref{sec:tree-uni-ecc}. In any case, the uni-ecc property enables us to derive an upper bound on the center-shift in Theorem~\ref{thm:center-shift-2side}.

\begin{theorem}\label{thm:center-shift-2side}
Let $G$ be a graph that satisfies the uni-ecc property and has center $C_G$, and let $\varphi : G \rightarrow H$ be an $(A,B,C)$-quasi-isometry. Then the center-shift is bounded above by
\[ \left( A - \frac{1}{A} \right) \rad(H)
    + AB + \frac{B}{A} .\]
\end{theorem}

\begin{proof}
Let $g\in C_G$ and $v\in\varphi^{-1}(C_H)$ be appropriate vertices such that $d_G (g, v)$ equals the center-shift. In other terms,
\[ d_G (g, v) = d_G \big( C_G , \varphi^{-1}(C_H) \big) ,\]
which by Proposition~\ref{prop:ecc-diff} can also be written as
\[ d_G (g, v) \leq \ecc_G(v) - \ecc_G(g) .\]
Then for $\ecc_G(v)$, Lemma~\ref{lma:q-iso-ecc} implies that $\ecc_G(v) \leq A \cdot \big( \ecc_H \big( \varphi(v) \big) + B \big)$,
Similarly, for $\ecc_G(g)$ we have
$-\ecc_G(g) \leq - (1/A) \cdot \big( \ecc_H \big( \varphi(g) \big) - B \big)$.
These two inequalities combine into:
\begin{align*}
d_G \big( C_G , \varphi^{-1}(C_H) \big) &= d_G (g, v) \\
&\leq 
\ecc_G(v) - \ecc_G(g)\\
&\leq
A \big( \ecc_H \big( \varphi(v) \big) + B \big)
 - \frac{1}{A} \big( \ecc_H \big( \varphi(g) \big) - B \big)\\
&= A \cdot \ecc_H \big( \varphi(v) \big) - \frac{1}{A} \cdot \ecc_H \big( \varphi(g) \big)
 + AB + \frac{B}{A} .
\end{align*}
Our assumption $v\in\varphi^{-1}(C_H)$ means that $\varphi(v)\in C_H$, so $\ecc_H \big( \varphi(v) \big) \leq \ecc_H \big( \varphi(g) \big)$.
Hence finally,
\[ d_G \big( C_G , \varphi^{-1}(C_H) \big) \leq
   \left( A - \frac{1}{A} \right) \cdot \ecc_H \big( \varphi(v) \big) 
   + AB + \frac{B}{A} .\]
Since $\ecc_H \big( \varphi(v) \big) = \rad(H)$, the proof is complete.
\end{proof}

Earlier, Proposition~\ref{prop:pgrf-path-leq} showed that $d_H \big( \varphi(v_1), \varphi(v_2) \big) \leq d_G (v_1, v_2)$, so the distance inequality of a partition-graph is more specific than a quasi-isometry. Such a ``one-sided quasi-isometry'' of a partition-graph yields a slightly more specific bound as demonstrated in Corollary~\ref{cor:center-shift-1side}. Its proof is omitted, as it is almost identical to the proof of Theorem~\ref{thm:center-shift-2side}.

\begin{corollary}\label{cor:center-shift-1side}
Let $G$ be a graph that satisfies the uni-ecc property and has center $C_G$, and let $\varphi : G \rightarrow H$ be a mapping that satisfies
$(1/A) \cdot d_G (x,y) - B
   \leq d_H \big( \varphi(x), \varphi(y) \big)
   \leq d_G (x,y)$.
Then the center-shift is bounded above by 
$(A-1) \cdot \rad(H) + AB$.
\end{corollary}

\subsection{Trees and Uni-Ecc Property}\label{sec:tree-uni-ecc}

Proposition~\ref{prop:ecc-diff} from earlier can be expressed as $d_G(C_G,v) \geq \ecc(v) - \rad(G)$ for any $v\in G$. Hence, the uni-ecc property is in fact a strong property that is not satisfied by most graphs. Theorem~\ref{thm:tree-uni-ecc} proves that trees satisfy the uni-ecc property, while Figure~\ref{fig:uni-ecc-counter-eg} shows a chordal graph that does not satisfy it.

\begin{figure}[t]
  \centering
  \includegraphics[scale=0.7]{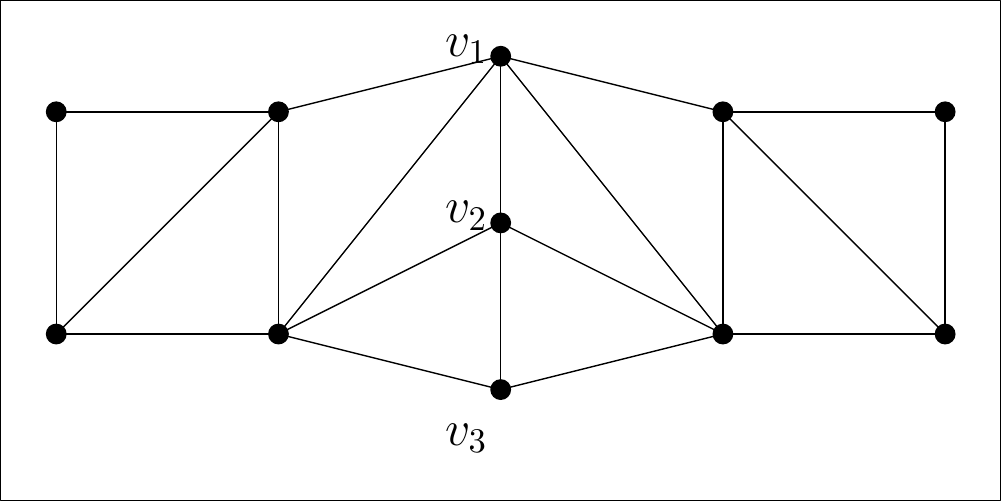}
  \caption{A chordal graph that does not satisfy the uni-ecc property. The center is $\{ v_1 \}$, but $d(C_G,v_3) = 2 \neq \ecc(v_3) - \rad(G) = 3 - 2$.}
  \label{fig:uni-ecc-counter-eg}
\end{figure} 

The proof of Theorem~\ref{thm:tree-uni-ecc} requires Lemma~\ref{lma:cent-ecc-wit} first. 

\begin{lemma}\label{lma:cent-ecc-wit}
Let $v$ be any vertex in a tree $T$, and let $c \in C_T$ be the center-vertex such that $d(v,c) = d(v,C_T)$. Then $c$ has an ecc-wit $w$ such that the path $v,\ldots w$ passes through $c$.
\end{lemma}

\begin{proof}
A tree's center consists of either one vertex or two adjacent vertices, so we can split the reasoning into two cases.

Suppose the center consists of just one vertex $c$. Since $c$ lies on the mid-point of a diameter-path, it must have at least two ecc-wits. Hence, no matter where our chosen $v$ is, there always exists at least one ecc-wit $w$ of $c$ such that the path $v,\ldots w$ passes through $c$.

Next, suppose the center consists of two adjacent vertices $c_1$ and $c_2$. Note that in this case, it is possible for a center-vertex to have only one ecc-wit. Nonetheless, the path between $c_1$ and any of its ecc-wits must pass through $c_2$, and \emph{vice versa}. Hence, given any $v$, let $c_1$ be the center-vertex that is closer to $v$ without loss of generality. Then, the path between $c_1$ and every ecc-wit of $c_1$ must pass through $c_2$, so consequently the path between $v$ and every ecc-wit of $c_1$ must pass through $c_2$. Therefore, there always exists at least one ecc-wit $w$ of $c_1$ such that the path $v,\ldots w$ passes through $c_1$.
\end{proof}

\begin{theorem}\label{thm:tree-uni-ecc}
In a tree $T$ with center $C_T$, every $v\in T$ satisfies $d(C_T, v) = \ecc(v) - \rad(T)$.
\end{theorem}

\begin{proof}
For every vertex $v$, let $c$ be its corresponding center-vertex such that $d(v,c) = d(v,C_T)$. Then by Lemma~\ref{lma:cent-ecc-wit}, there is an ecc-wit $w$ of $c$ such that the path $v,\ldots w$ passes through $c$. This means that $d(v,w) = d(v,c) + d(c,w)$.

We then show that $w$ is not only an ecc-wit of $c$, but also an ecc-wit of $v$. For this we need to establish that $\forall x\in T : d(v,x) \leq d(v,w)$. We look at two cases for $x$.

Firstly, suppose the path $v,\ldots x$ passes through $c$. Then
\begin{align*}
      d(v,x)
   &= d(v,c) + d(c,x) \\
&\leq d(v,c) + d(c,w) &\text{($w$ is ecc-wit of $c$)}\\
   &= d(v,w) .
\end{align*}

Secondly, suppose $v,\ldots x$ passes through $c$. Then consider the two paths $v,\ldots c$ and $v,\ldots x$. Let $v'$ be the last common vertex on these paths. Hence,
\begin{align*}
      d(v,x)
   &= d(v,v') + d(v',x)\\
   &= d(v,v') + d(c,x) - d(c,v') \\
&\leq d(v,v') + d(c,w) - d(c,v') &\text{($w$ is ecc-wit of $c$)}\\
   &< d(v,v') + d(c,w) + d(c,v') &\text{(adding positive $d(c,v')$ twice)}\\
   &= d(v,w) .
\end{align*}
Combining the two cases, we can conclude that $w$ is an ecc-wit of $v$. This means that
\begin{align*}
   \ecc(v) - \rad(T)
&= \ecc(v) - \ecc(c)\\
&= d(v,w) - d(c,w)\\
&= d(v,c)\\
&= d(v,C_T),
\end{align*}
which is the statement of Theorem~\ref{thm:tree-uni-ecc}.
\end{proof}

Theorem~\ref{thm:center-shift-2side} and Corollary~\ref{cor:center-shift-1side} already ensure that any quasi-isometric simplification of a tree has a bounded center-shift. However, this upper bound is a function of the radius of the tree, so we want to further investigate if the center-shift can be bounded by a constant under particular quasi-isometric simplifications of trees.

\section{Outward-Contraction and Centers of Trees}
\label{sec:outward-contraction}

This section studies \emph{partition-trees}, which means partition-graphs on trees. We present the \emph{outward-contraction} algorithm, which constructs a specific type of partitions on trees. We then show that outward-contraction always produces partition-trees with center-shift zero.

Corollary~\ref{cor:center-shift-1side} together with Theorem~\ref{thm:tree-uni-ecc} showed that the center-shift of any partition-tree is bounded by a function of its radius. However, when the partition-tree has a large radius, the center-shift as a numerical value can still become arbitrarily large, which we discuss in Example~\ref{eg-comb} below.

\begin{example}\label{eg-comb}
Given a non-negative integer $k$, this example demonstrates that there is a tree $G_k$ with a partition, such that the centre-shift is $k$. The tree $G_k$ is built from three smaller trees described below.

Let $H$ and $M_k$ be the trees as shown in Figure~\ref{fig:tree-comb-0}, and let $P_{3k}$ be the path-graph on $3k$ vertices labelled as $v_1,\ldots v_{3k}$.

Then $G_k$ is obtained by linking $v_{3k}$ in $P_{3k}$ and $a$ in $H$ with an edge, and linking $x_k$ in $M_k$ and $b$ in $H$ with another edge. The result is the right-most graph in Figure~\ref{fig:tree-comb-0}.

\begin{figure}[t]
  \centering
  \includegraphics[scale=0.8]{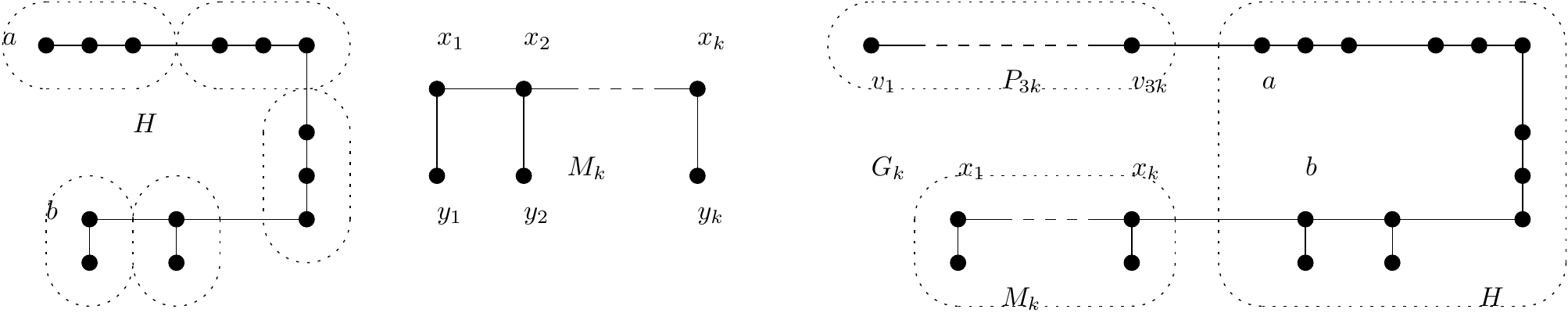}
  \caption{For Example~\ref{eg-comb}: the tree $H$ and its partition of interest, the tree $M_k$, and the combined tree of interest $G_k$.}
  \label{fig:tree-comb-0}
\end{figure}

The partition on $G_k$ is as follows. On the $P_{3k}$ part, we group $v_{3j-2}$, $v_{3j-1}$ and $v_j$ into one super-vertex for every $1\leq j\leq k$. On the $H$ part, we group the vertices as depicted in Figure~\ref{fig:tree-comb-0}. On the $M_k$ part, we group $\{ x_j, y_j \}$ for every $1\leq j\leq k$. Figure~\ref{fig:tree-comb} demonstrates this construction for $k=2$ and $k=3$.

\begin{figure}[t]
  \centering
  \includegraphics[scale=0.8]{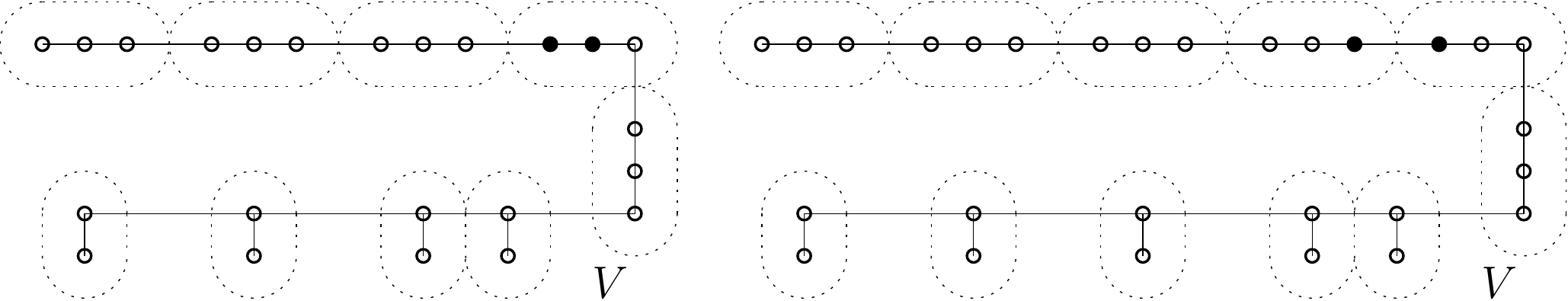}
  \caption{Example~\ref{eg-comb} with $k=2$ and $k=3$.}
  \label{fig:tree-comb}
\end{figure}

Finally we compute the center-shift by locating the center of $G_k$ and the center of the partition-tree $\ol{G}_k$. In $G_k$, the path from $v_1$ to $y_1$ is a diameter-path, and the center of $G_k$ coincides with the center of this path. This path has $4k+12$ vertices, so if we re-label its vertices from one end to the other as $1,\ldots 4k+12$, then its center is $\{ 2k+6, 2k+7\}$.
On the other hand, by counting the super-vertices as in Figure~\ref{fig:tree-comb}, one can easily see that the center of $\ol{G}_k$ consists only one super-vertex $\{ 3k+7, 3k+8, 3k+9 \}$.
Overall, the center-shift is $(3k+7)-(2k+7) = k$.
\end{example}

Although Example~\ref{eg-comb} showed that the center-shift of any arbitrary partition-tree is not always a small number, we now present the \emph{outward-contraction} algorithm, which generates a partition in a particular way, such that the resulting partition-tree has center-shift zero. Outward-contration takes any unrooted tree as its input, and designates an arbitrary vertex as the root. With the root designated, the following definitions will become useful.

\begin{definition}[Level-function]\label{def:tree-level}
For every vertex $v$ in a tree $T$ with root $r$, the \emph{level} of $v$ is defined to be $\lev(v) = d(v,r)$.
\end{definition}

\begin{definition}[Monotone path]\label{def:monotone-path}
In a rooted tree, a simple path $v_1,\ldots v_k$ is called \emph{monotone} when either:
\begin{itemize}
  \item $\lev(v_i) < \lev(v_{i+1})$ for all $1 \leq i < k$, or
  \item $\lev(v_i) > \lev(v_{i+1})$ for all $1 \leq i < k$.
\end{itemize}
\end{definition}

\begin{proposition}\label{prop:at-most-1-turnpt}
In a rooted tree, a simple path either is monotone, or can be decomposed into two monotones simple paths.
\end{proposition}

\begin{proof}
Consider the unique simple path between any two vertices $v_1$ and $v_2$ in a tree. If this path is monotone, then the statement already holds.

On the other hand, suppose this path is not monotone. Then consider the following two paths $r, \ldots v_1$ and $r, \ldots v_2$. Let $u$ be the last common vertex of these two paths. Then $v_1, \ldots u, \ldots v_2$ is the unique path between $v_1$ and $v_2$, and indeed this path can be decomposed into two monotones simple paths.
\end{proof}

In the proof of Proposition~\ref{prop:at-most-1-turnpt} above, it is convenient to refer to the vertex $u$ as a \emph{turning-point}, which is formalised in Definition~\ref{def:turn-pt}.

\begin{definition}[Turning-point]\label{def:turn-pt}
The \emph{turning-point} of a non-monotone path $v_1, \ldots v_k$ is the vertex $v_i$ (with $2\leq i\leq k-1$) such that $\lev(v_{i-1}) > \lev(v_i)$ and $\lev(v_i) < \lev(v_{i+1})$. 
\end{definition}

Hence, Proposition~\ref{prop:at-most-1-turnpt} can also be phrased as ``every path in a rooted tree has at most one turning-point''.

\begin{definition}[Outward-contraction]\label{def:outward-contraction}
The \emph{outward-contraction} algorithm takes a tree as input, and constructs a partition as follows.
\begin{enumerate}
  \item It designates an arbitrary vertex as a root.
  \item For every vertex $v$ with $\lev(v)$ an even number, let $N_{\downarrow}(v)$ be the set of neighbors $u$ of $v$ such that $\lev(v) < \lev(u)$. Then outward-contraction assigns each $v$ and $N_{\downarrow}(v)$ into one super-vertex.
  \item The output is the corresponding partition-tree.
\end{enumerate}
\end{definition}

In a partition-tree produced by outward-contraction, each super-vertex has the following three possible forms. Let $v$ be a vertex with even $\lev(v)$:
\begin{itemize}
  \item If $N_{\downarrow}(v)$ is empty, then the super-vertex consists of just $v$. Hence the diameter of the super-vertex is zero.
  \item If $N_{\downarrow}(v)$ has only one vertex, then the super-vertex consists of only two adjacent vertices. Hence the diameter of the super-vertex is one.
  \item If $N_{\downarrow}(v)$ has at least two vertices, then the super-vertex has diameter two.
\end{itemize}
Consequently, the partitions produced by outward-contraction are 2-sharp and 0-coarse.
Figure~\ref{fig:root-contr-med-eg} shows an example of outward-contraction, where the designated vertex is marked by a square.

\begin{figure}[t]
  \centering
  \includegraphics[scale=0.8]{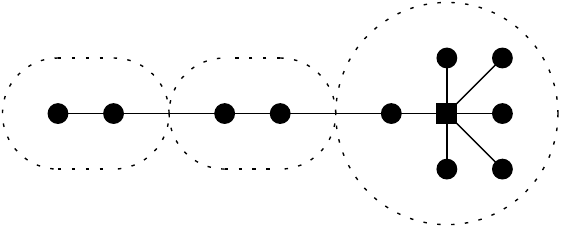}
  \caption{An example of outward-contraction starting from the square vertex. Note that the center is preserved but the median is not.}
  \label{fig:root-contr-med-eg}
\end{figure}

The center of a tree always lies on a diameter-path, and hence is the same as the center of any diameter-path of the tree~\cite{wu-chao-04}. This allows us to reduce the problem of locating the center of a tree into the simpler problem of locating the center of a path.

On the path-graph $P_n$, a partition can be expressed as a sequence of natural numbers that represent the sizes of the super-vertices from left to right. The numbers in such a sequence sum to $n$, so sequences like these are in fact integer compositions. Thus a partition-path $P_k$ of the path-graph $P_n$ can be viewed an integer composition of $n$ with length $k$. The constituent numbers of the integer composition $P_k$ is denoted by $P_k[1], P_k[2], \ldots P_k[k]$.

\begin{definition}
Let $w$ be an integer composition of length $k$. Then the set of \emph{center-indices} is
\begin{itemize}
  \item $\{(k+1)/2\}$   when $k$ is odd, or
  \item $\{k/2,~k/2+1\}$ when $k$ is even.
\end{itemize}
Furthermore:
\begin{itemize}
  \item The \emph{center-sum} $\sigma$ is $\sum w[i]$, for $i$ in the set of center-indices.
  \item The \emph{left-sum} $\lambda$ is $\sum w[i]$, for $i$ smaller than all the center-indices.
  \item The \emph{right-sum} $\rho$ is $\sum w[i]$, for $i$ larger than all the center-indices.
\end{itemize}
\end{definition}

\begin{example}
Consider $w=$ \tw{332231}, which represents a partition on $P_{14}$. (In the rest of the paper, we write integer compositions in \tw{typewriter} \tw{font} to aid clarity.) The center-indices of $w$ are 3 and 4, so the center-sum $\sigma = w[3] + w[4] = 2+2 = 4$. Its left-sum is $\lambda = 6$, its right-sum is $\rho = 4$, and $|\lambda - \rho| = 2$.
Now, since $\sigma \geq |\lambda - \rho|$, we can straightaway conclude that the center-shift is zero. 
\end{example}

Now, by executing leaf-removal (Note~\ref{note:leaf-removal}) on $P_k$ and $P_n$, one can observe the useful fact in Lemma~\ref{lma:int-comp-center-shift} below.

\begin{lemma}\label{lma:int-comp-center-shift}
Let $P_k$ be a partition-path of $P_n$, and let $w$ be the integer composition that represents $P_k$. Also, let $\sigma$, $\lambda$ and $\rho$ respectively denote the center-sum, left-sum and right-sum of $w$. Then the center-shift is 0 if $\sigma \geq |\lambda - \rho|$, or $\lceil \big(|\lambda - \rho| - \sigma\big)/2 \rceil$ otherwise.
\end{lemma}

\begin{proof}
Let $C_n$ and $C_k$ be the centers of $P_n$ and $P_k$, respectively.

We picture $P_n$ as consisting of three subgraphs, each connected, called Segments L, C and R. Segment C corresponds to $\varphi^{-1}(C_k)$, and let $\sigma$ be the number of vertices in Segment C. Then, let Segments L and R be the two connected components after removing Segment C from $P_n$. Let $\lambda$ and $\rho$ respectively denote the number of vertices in Segments L and R, and assume $\lambda \leq \rho$ without loss of generality.

Now we use the leaf-removal algorithm (Note~\ref{note:leaf-removal}) to locate $C_n$, the center of $P_n$. On a path-graph, one iteration of leaf-removal is the same as removing both endpoints. Later we will calculate its distance in $P_n$ between $C_n$ and Segment C.

We first carry out $\lambda$ iterations. Every vertex in Segment L is removed, so we are left with all the vertices in Segment C, plus potentially some vertices in Segment R if $\lambda < \rho$. Let Segment R-L denote the left-over $\rho - \lambda$ vertices on Segment R. From here there are three possible cases.

\begin{enumerate}
  \item When Segments C and R-L have equal length, the center of $P_n$ is made up of the right-most vertex in Segment C and the left-most vertex in Segment R-L, so the center-shift is zero.
  \item When Segment C is longer than R-L, the center of $P_n$ lies in Segment C, so the center-shift is zero.
\end{enumerate}
These two cases combine to prove the first part of Lemma~\ref{lma:int-comp-center-shift}: the center-shift is 0 when $\sigma \geq |\lambda - \rho|$. In contrast, the final case involves more effort to quantify the non-zero center-shift.

\begin{enumerate}[start=3]
  \item When Segment R-L is longer than C, the center of $P_n$ falls in Segment R-L, and the non-zero center-shift is the distance between $C_n$ and Segment C. This distance is the same as the distance between the right-most vertex in Segment C and the left-most vertex of $C_n$.
\end{enumerate}
The right-most vertex in Segment C has index $\sigma$. On the other hand, on a path of length $n$, the index of the left-most center-vertex is $\lceil n/2 \rceil$, so the left-most vertex of $C_n$ has index $\lceil (\sigma + |\lambda - \rho|)/2 \rceil$.
Finally, we can derive the center-shift by making the following subtraction:
\[
  \left\lceil \frac{\sigma + |\lambda - \rho|}{2} \right\rceil - \sigma
= \left\lceil \frac{\sigma + |\lambda - \rho|}{2} - \sigma \right\rceil
= \left\lceil \frac{|\lambda - \rho| - \sigma}{2} \right\rceil ,\]
and this proves the second part of Lemma~\ref{lma:int-comp-center-shift}.
\end{proof}

\begin{theorem}\label{thm:tree-center-preserve}
Outward-contraction always produces a partition-tree with center-shift zero.
\end{theorem}

\begin{proof}
Let $T$ be the input tree. Outward-contraction arbitrarily designates a vertex, and generates a partition $\mathcal{P}$ on $T$.

Consider a diameter-path $D$ in $T$. Let $\mathcal{P}_D = \{ X \cap D \mid X \in \mathcal{P} \}$ be the restriction of $\mathcal{P}$ to $D$. In the rest of the proof, we focus only on the super-vertices in $\mathcal{P}_D$. Since $D$ is a path, we refer to the sizes of its super-vertices as elements of an integer composition.

Earlier, Proposition~\ref{prop:at-most-1-turnpt} showed that every path in a rooted tree has at most one turning-point. If $D$ does have a turning-point, then the super-vertex containing the turning-point has size either one or three. In the integer composition that represents $\mathcal{P}_D$, such a super-vertex is represented by a \tw{1} or a \tw{3}.
On the other hand, the endpoints of $D$ are contained in super-vertices with size one or two, and hence a \tw{1} or a \tw{2} in the integer composition. Meanwhile, all the remaining super-vertices that contain neither the turning-point nor endpoints always have size two.
We now consider leaf-removal (Note~\ref{note:leaf-removal}) on $D$, which leads to two possible cases.

\begin{enumerate}
  \item Suppose leaf-removal does not encounter the turning-point throughout the execution.
\end{enumerate}
This occurs when $D$ has no turning-point, or when the super-vertex containing the turning-point is in the center of $\ol{D}$ (the partition-path of $D$ induced by $\mathcal{P}_D$). The center of $\ol{D}$ contains either one or two super-vertices, and at most one of these super-vertices contains the turning-point.

If the center of $\ol{D}$ contains only one super-vertex, then either this super-vertex contains the turning-point and has size one or three, or it does not contain a turning point and has size two. Overall, the center-sum of $\mathcal{P}_D$ is one, two or three.

If the center of $\ol{D}$ contains two super-vertices, then these super-vertices correspond to the following possible integer compositions: \tw{12}, \tw{21}, \tw{32}, \tw{23} or \tw{22}. The first four occur when one of these super-vertices in the center contains the turning-point, while the last composition \tw{22} occurs when neither super-vertex in the center contains the turning-point.

Now, the possible center-sum $\sigma$ ranges from one to five. Then there are four further subcases depending on whether each endpoint of $D$ is a \tw{1} or a \tw{2}. These subcases are listed in Table~\ref{table:1} alongside their respective $|\lambda - \rho|$ values. The center-super-vertices are marked by \tw{[]}, and the dots all stand for \tw{2}.

\begin{table}[t]
  \caption{For Case 1 of the proof of Theorem~\ref{thm:tree-center-preserve}.}
  \label{table:1}
  \centering
\begin{tabular}{c|c}
  subcase & $|\lambda - \rho|$ \\\hline
  \tw{12...[]...21} & 0 \\
  \tw{22...[]...22} & 0 \\
\end{tabular}
\quad
\begin{tabular}{c|c}
  subcase & $|\lambda - \rho|$ \\\hline
  \tw{12...[]...22} & 1 \\
  \tw{22...[]...21} & 1 \\
\end{tabular}
\end{table}

As $\sigma \geq |\lambda - \rho|$ in all possible cases, by Theorem~\ref{lma:int-comp-center-shift} the center-shift is always zero.

\begin{enumerate}[start=2]
  \item Suppose leaf-removal encounters the turning-point of $D$ at some point during the execution.
\end{enumerate}
Then the turning-point is not in any super-vertex of the center of $\ol{D}$, so the possible values of the center-sum are two (one super-vertex in the center) and four (two super-vertices in the center).

Without loss of generality, assume that the super-vertex containing the turning-point is on the left-hand side of the center of $\ol{D}$. Depending on whether each endpoint is a \tw{1} or a \tw{2}, as well as whether the super-vertex containing the turning-point is a \tw{1} or a \tw{3}, there are eight subcases listed alongside the corresponding $|\lambda - \rho|$ values in Table~\ref{table:2}. Again, the super-vertices in the center of $\ol{D}$ are marked by \tw{[]}, and the dots all stand for \tw{2}.

\begin{table}[t]
  \caption{For Case 2 of the proof of Theorem~\ref{thm:tree-center-preserve}.}
  \label{table:2}
  \centering
\begin{tabular}{c|c}
  subcase & $|\lambda - \rho|$ \\\hline
  \tw{12..1..[].....21} & 1 \\
  \tw{22..1..[].....22} & 1 \\
  \tw{12..1..[].....22} & 2 \\
  \tw{22..1..[].....21} & 0 \\
\end{tabular}
\quad
\begin{tabular}{c|c}
  subcase & $|\lambda - \rho|$ \\\hline
  \tw{12..3..[].....21} & 1 \\
  \tw{22..3..[].....22} & 1 \\
  \tw{12..3..[].....22} & 0 \\
  \tw{22..3..[].....21} & 2 \\
\end{tabular}
\end{table}

As $\sigma \geq |\lambda - \rho|$ in all cases, by Lemma~\ref{lma:int-comp-center-shift}, the center-shift is always zero.
\end{proof}

\section{Vertex-Weighted Partition-Trees and Medians}\label{sec:vertex-weights}

Although outward-contraction preserves the center of a tree (Theorem~\ref{thm:tree-center-preserve}), it does not always preserve the median. Figure~\ref{fig:root-contr-med-eg} shows a counter-example.

Nevertheless, partition-trees can still preserve the median by taking the sizes of the super-vertices into account. This brings us to define vertex-weighted graphs and the vertex-weighted distance-sum, which were also used in~\cite{hakimi-64}.

\begin{definition}[Vertex-weight]\label{def:vert-wei}
A \emph{vertex-weighted graph} $G$ is a graph with a \emph{vertex-weight} function $f:V(G)\rightarrow\mathbb{R}^+$.
In addition, the \emph{weight of a vertex-subset} $S$, written as $f(S)$, is defined to be the sum of the weights of all $v\in S$.

Then in a vertex-weighted graph $G$, the \emph{vertex-weighted distance-sum} of each vertex $x$ is defined to be $\dsw(x) = \sum_{v\in G} d(x,v) \cdot f(v)$.
Then the \emph{median} of a vertex-weighted graph is the set of vertices that minimize the distance-sum function. 
\end{definition}

Definition~\ref{def:vert-wei} gives rise to Lemma~\ref{lma:tree-vw-ds-cut} and Corollaries~\ref{cor:vert-wei-trees-1},~\ref{cor:vert-wei-trees-2} and~\ref{cor:vert-wei-trees-med}. These then lead to the ultimate Theorem~\ref{thm:median-preserved}.

\begin{lemma}\label{lma:tree-vw-ds-cut}
Let $T$ a vertex-weighted tree with $f$ as the vertex-weight function, and let $x$ and $y$ be adjacent vertices in $T$. Furthermore, let $S_x$ denote the set of vertices that pass through $x$ in order to reach $y$; the set $S_y$ is defined symmetrically. (See Figure~\ref{fig:tree-ds-cut}.)

\begin{figure}[t]
  \centering
  \includegraphics[scale=0.7]{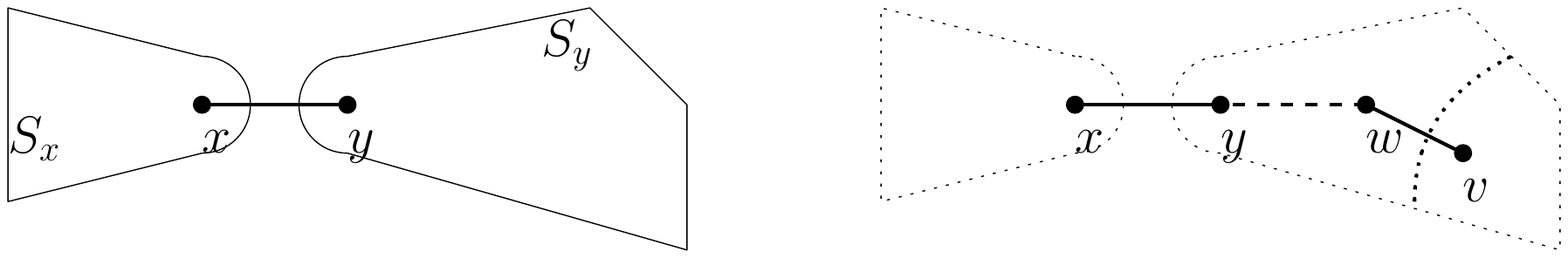}
  \caption{Left: the two subtrees $S_x$ and $S_y$ separated by the edge $\{x,y\}$. Right: for $v\in S_y$, let $w$ be its neighbor on its path to $y$. Then $S_w$ and $S_v$ are the two subtrees separated by $\{w,v\}$. Every vertex to the left of the dotted arc between $w$ and $v$ belongs to $S_w$, while every vertex to the right belongs to $S_v$.}
  \label{fig:tree-ds-cut}
\end{figure}

Then $\dsw(x) + f(S_x) = \dsw(y) + f(S_y)$. This also implies that $\dsw(x) < \dsw(y)$ if and only if $f(S_y) < f(S_x)$.
\end{lemma}

\begin{proof}
We begin by deriving $\dsw(x)$:
\[ \dsw(x)
  = \sum_{v\in G} d(x,v) \cdot f(v)
  = \sum_{v\in S_x} \big[ d(x,v) \cdot f(v) \big]
 + \sum_{v\in S_y} \big[ d(x,v) \cdot f(v) \big] .\]
Noting that $d(x,v) = d(x,y) + d(y,v) = 1 + d(y,v)$. the second sum can be routinely transformed into:
\[ \sum_{v\in S_y} \big[ d(x,v) \cdot f(v) \big] = f(S_y) + \sum_{v\in S_y} \big[ d(y,v) \cdot f(v) \big].\]
Hence,
\[ \dsw(x) = \sum_{v\in S_x} \big[ d(x,v) \cdot f(v) \big] + f(S_y) + \sum_{v\in S_y} \big[ d(y,v) \cdot f(v) \big] .\]
The exact same argument also yields:
\[ \dsw(y) =
  \sum_{v\in S_x} [ d(x,v) \cdot f(v) ]
  + f(S_x) + \sum_{v\in S_y} [ d(y,v) \cdot f(v) ] .\]
Therefore, after subtracting these two equations and rearranging, we obtain the lemma's statement.
\end{proof}

\begin{corollary}\label{cor:vert-wei-trees-1}
Let $T$ a vertex-weighted tree with $f$ as the vertex-weight function, and let $x, y, S_x, S_y$ be defined in the same way as in Lemma~\ref{lma:tree-vw-ds-cut}. 

Then $\dsw(x) < \dsw(y)$ implies that for all $v\in S_y : \dsw(x) < \dsw(v)$.
\end{corollary}

\begin{proof}
For every $v\in S_y$, let $w\in S_y$ be the neighbor of $v$ that lies on the path between $y$ and $v$. Define $S_w$ to be the set of vertices whose paths to $v$ pass through $w$, and similarly for $S_v$. (See Figure~\ref{fig:tree-ds-cut}.)

Then $S_x\subset S_w$ and $S_v\subset S_y$. Since all the vertex-weights are positive, these containments imply
\begin{equation}\label{eq:cor-dsw}
  f(S_x) < f(S_w) \text{~and~} f(S_v) < f(S_y).
\end{equation}

Due to the premise $\dsw(x) < \dsw(y)$, Lemma~\ref{lma:tree-vw-ds-cut} implies that $f(S_y) < f(S_x)$. Combining this with (\ref{eq:cor-dsw}) leads to $f(S_v) < f(S_w)$.

Then by Lemma~\ref{lma:tree-vw-ds-cut} again, we have $\dsw(w) < \dsw(v)$. Finally, using routine induction on $d(y,v)$, we can extend the observation above to the entire $S_y$, and conclude that $\dsw(x) < \dsw(v)$ for all $v\in S_y$.
\end{proof}

\begin{corollary}\label{cor:vert-wei-trees-2}
Let $T$ a vertex-weighted tree with $f$ as the vertex-weight function, and let $x, y, S_x, S_y$ be defined in the same way as in Lemma~\ref{lma:tree-vw-ds-cut}. 

Then $\dsw(x) = \dsw(y)$ implies that $\{x,y\}$ is the median.
\end{corollary}

\begin{proof}
By Lemma~\ref{lma:tree-vw-ds-cut}, $\dsw(x) = \dsw(y)$ is equivalent to $f(S_x) = f(S_y)$. This means that $f(S_x) = f(T)/2$, where $T$ is the entire tree.
Without loss of generality, consider a vertex $v\in S_x$ such that $v\sim x$. With respect to the edge $\{v,x\}$, let $A_x$ be the subtree on the side of $x$, and $A_v$ the subtree on the side of $v$.

Now, as $A_x = S_x \cup \{x\}$, we have $f(A_v) < f(A_x)$, and therefore $\dsw(v) > \dsw (x)$. Finally, apply Corollary~\ref{cor:vert-wei-trees-1} to every such $v$ in both $S_x$ and $S_y$, we conclude that $\dsw(x)$ and $\dsw(y)$ are indeed the minimum. Therefore $\{x,y\}$ is the vertex-weighted median.
\end{proof}

\begin{corollary}\label{cor:vert-wei-trees-med}
The median of a vertex-weighted tree consists of either one vertex or two adjacent vertices.
\end{corollary}

\begin{proof}
Firstly, it is easy to construct examples of vertex-weighted trees with medians being a single vertex or two adjacent vertices.
Secondly, it suffices to show that in a tree $T$ with vertex-weight function $f$, any two vertices in the vertex-weighted median are adjacent. This not only implies that the vertex-weighted median is connected, but also excludes the possibility of the median having three or more vertices. 

Let $v_1$ and $v_p$ be vertices with the minimum $\dsw$ value, and suppose they are separated by a path $v_2,\ldots v_{p-1}$.
This is shown in Figure~\ref{fig:tree-med-thm}, where $S_1,\ldots S_p$ indicate the subtrees of the vertices $v_1,\ldots v_p$.

\begin{figure}[t]
  \centering
  \includegraphics[scale=0.7]{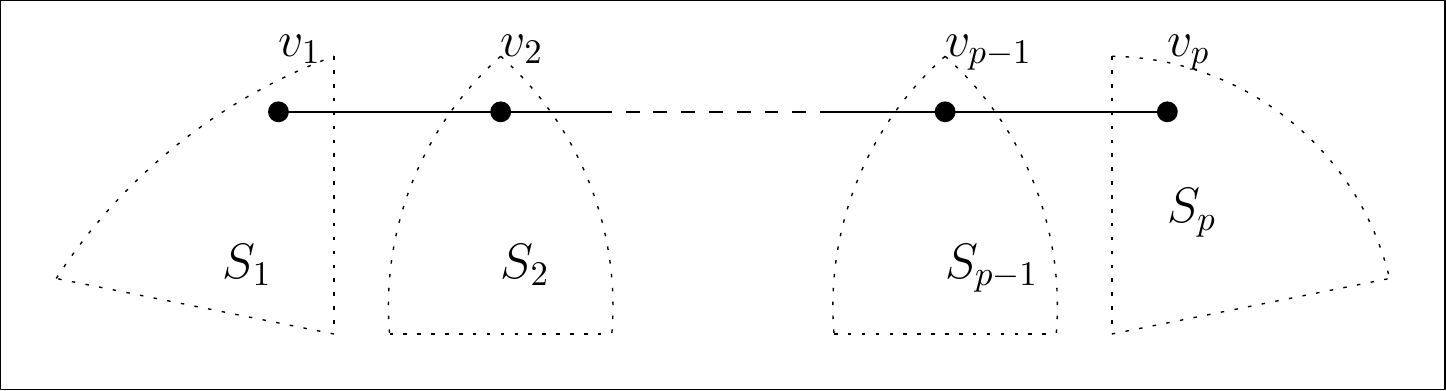}
  \caption{For the proof of Corollary~\ref{cor:vert-wei-trees-med}.}
  \label{fig:tree-med-thm}
\end{figure}

Consider the adjacent vertices $v_1$ and $v_2$. Since $v_1$ has the minimum $\dsw$ value, $\dsw(v_1) \leq \dsw(v_2)$. Then by Lemma~\ref{lma:tree-vw-ds-cut},
\begin{equation}\label{eq:appendix-1}
f(S_1) \geq \sum_{j=2}^{p} f(S_j) .
\end{equation}
Similarly, consider the adjacent vertices $v_{p-1}$ and $v_p$. Since $v_p$ has the minimum $\dsw$ value, $\dsw(v_p) \leq \dsw(v_{p-1})$, and hence
\begin{equation}\label{eq:appendix-2}
  f(S_p) \geq \sum_{j=1}^{p-1} f(S_j) .
\end{equation}
Summing (\ref{eq:appendix-1}) and~(\ref{eq:appendix-2}) leads to
$0 \geq \sum_{j=2}^{p-1} f(S_j)$.
But by Definition~\ref{def:vert-wei}, the weights of vertices are all positive, so this is a contradiction.
Therefore, two vertices with the minimum $\dsw$ value must be adjacent, and hence the corollary holds.
\end{proof}

With these basics of vertex-weighted graphs in place, we move on to define how vertex-weights are incorporated into the framework of partition-graphs.

\begin{definition}
Given a partition on a graph $G$, the \emph{vertex-weighted partition-graph} $\ol{G}$ is defined as follows.
\begin{itemize}
  \item The vertices and edges of $\ol{G}$ are the same as in Definition~\ref{def:partition-graph}.
  \item The weight of each vertex $X$ in $\ol{G}$ is the cardinality of its corresponding subset of $G$.
\end{itemize}
\end{definition}

Now we can state and prove the main theorem of this section. On the notation, the super-vertices in $\ol{G}$ are denoted using capital letters, and the distance-sum of a super-vertex $X$ in $\ol{G}$ is denoted by $\dsw(X)$. Since there is little chance of ambiguity, we overload the notation for convenience.

\begin{theorem}\label{thm:median-preserved}
Let $T$ be a tree, and let $\ol{T}$ denote the vertex-weighted partition-graph induced by \emph{any} partition on $T$. Then every super-vertex in the median of $\ol{T}$ contains a vertex in the median of $T$.
\end{theorem}

\begin{proof}
Since the vertex-weighted $\ol{T}$ is still a tree, its median is either a single super-vertex of two adjacent super-vertices (Corollary~\ref{cor:vert-wei-trees-med}), so we have the following two cases.

For the first case, let $X$ be the only super-vertex in the median of $\ol{T}$. Then by definition, for every neighbor $Y$ of $X$, $\dsw(X) < \dsw(Y)$.
Let $\ol{S}_X$ denote the set of super-vertices that have to pass through $X$ in order to reach $Y$, and let $\ol{S}_Y$ be the analogous counterpart.
Then by Lemma~\ref{lma:tree-vw-ds-cut}, $f(\ol{S}_X) > f(\ol{S}_Y)$.

Let $x\in X$ and $y\in Y$ such that $x$ and $y$ are adjacent in $T$. Define $S_x$ to be the set of vertices whose paths to $y$ pass through $x$, and define $S_y$ analogously.
Now observe that $f(\ol{S}_X) = |S_x|$ and $f(\ol{S}_Y) = |S_y|$.
This means that $|S_x| > |S_y|$ and hence $\dsw(x) < \dsw(y)$.
By Corollary~\ref{cor:vert-wei-trees-1}, every vertex $v\in S_y$ has a bigger distance-sum than $x$. Since every vertex not in $X$ does not have the minimum distance-sum, so the median-vertices of $T$ must be inside $X$.

For the second case, let $X$ and $Y$ be the two adjacent super-vertices in the median of $\ol{T}$. Let $x\in X$ and $y\in Y$ be the corresponding adjacent vertices in $T$.
In addition, define $\ol{S}_X$, $\ol{S}_Y$, $S_x$ and $s_y$ as before. Now $\dsw(X) = \dsw(Y)$ implies $f(\ol{S}_X) = f(\ol{S}_Y)$.
This further means that $|S_x| = |S_y|$ and hence $\dsw(x) = \dsw(y)$.
Finally, using Corollary~\ref{cor:vert-wei-trees-2}, $x$ and $y$ are the two median-vertices of $T$.
\end{proof}

\section{Conclusion}\label{sec:conclusion}

Quasi-isometries capture the general notion of distance-approximation, and this paper introduced them into the field of graph simplification. The goals of quasi-isometric graph simplification were listed in Section~\ref{sec:setup}. After laying down the basics of quasi-isometries in Section~\ref{sec:q-iso}, we presented some constructions of quasi-isometric graph simplifications, and evaluated them against the goals outlined in Section~\ref{sec:setup}.

The first construction called \emph{MIS-derived graphs} was presented in Section~\ref{sec:q-iso-ind-set}. This construction is based on maximal independent sets. Though it has \emph{small quasi-isometry constants} and is mathematically interesting, it does not satisfy the \emph{compression} property in general, nor does it satisfy the \emph{retention} property for trees.

The main construction we studied was \emph{partition-graphs}. Section~\ref{sec:partition-graphs} showed that partition-graphs satisfy the first two goals of \emph{small quasi-isometry constants} and \emph{compression}, given suitable values of sharpness and coarseness. We then focused on trees, where partition-graphs satisfy the \emph{retention} property, as every partition-graph of a tree remains a tree (Corollary~\ref{cor:pgrf-tree-retention}). As for the \emph{preservation} property, Sections~\ref{sec:outward-contraction} and~\ref{sec:vertex-weights} presented constructions of partition-trees that preserve the center and the median, respectively. 

For future work, the most pressing step is to address the \emph{efficiency} criterion. This entails devising concrete algorithms that efficiently construct partition-graphs or other quasi-isometric graph simplifications. Meanwhile, as a graph can have multiple possible partition-graphs, one could explore the possibility of employing the theory of random graphs to investigate the ``average'' properties among all the possible partition-graphs of a given graph.






\bibliographystyle{plain}
\bibliography{ref-qiso-grf}

\begin{thebibliography}{10}

\bibitem{metric-embed-abraham-11}
Ittai Abraham, Yair Bartal, and Ofer Neiman.
\newblock Advances in metric embedding theory.
\newblock {\em Advances in Mathematics}, 228(6):3026--3126, 2011.

\bibitem{ahmed-20-spanner-survey}
Reyan Ahmed, Greg Bodwin, Faryad~Darabi Sahneh, Keaton Hamm, Mohammad
  Javad~Latifi Jebelli, Stephen Kobourov, and Richard Spence.
\newblock Graph spanners: {A} tutorial review.
\newblock {\em Computer Science Review}, 37, 2020.

\bibitem{graph-databases}
Renzo Angles and Claudio Gutierrez.
\newblock Survey of graph database models.
\newblock {\em ACM Computing Surveys}, 40(1):1--39, 2008.

\bibitem{bast-16}
Hannah Bast, Daniel Delling, Andrew Goldberg, Matthias M{\"u}ller-Hannemann,
  Thomas Pajor, Peter Sanders, Dorothea Wagner, and Renato~F. Werneck.
\newblock Route planning in transportation networks.
\newblock In {\em Algorithm Engineering}, pages 19--80. Springer, 2016.

\bibitem{benczur-15}
Andr\'as~A. Bencz\'ur and David~Ron Karger.
\newblock Randomized approximation schemes for cuts and flows in capacitated
  graphs.
\newblock {\em SIAM Journal on Computing}, 44(2):290--319, 2015.

\bibitem{bernstein-19}
Aaron Bernstein, Karl D{\"a}ubel, Yann Disser, Max Klimm, Torsten M{\"u}tze,
  and Frieder Smolny.
\newblock Distance-preserving graph contractions.
\newblock {\em SIAM Journal on Discrete Mathematics}, 33(3):1607–1636, 2019.

\bibitem{besta-18-survey}
Maciej Besta and Torsten Hoefler.
\newblock Survey and taxonomy of lossless graph compression and space-efficient
  graph representations.
\newblock arXiv:1806.01799, 2018.

\bibitem{bichot-13}
Charles‐Edmond Bichot and Patrick Siarry, editors.
\newblock {\em Graph Partitioning}.
\newblock Wiley, 2013.

\bibitem{char-19}
Panagiotis Charalampopoulos, Pawe\l{} Gawrychowski, Shay Mozes, and Oren
  Weimann.
\newblock Almost optimal distance oracles for planar graphs.
\newblock In {\em STOC}, pages 138--151, 2019.

\bibitem{cohen-03}
Edith Cohen, Eran Halperin, Haim Kaplan, and Uri Zwick.
\newblock Reachability and distance queries via 2-hop labels.
\newblock {\em SIAM Journal on Computing}, 32(5):1338--1355, 2003.

\bibitem{fan-22-contraction}
Wenfei Fan, Yuanhao Li, Muyang Liu, and Can Lu.
\newblock Making graphs compact by lossless contraction.
\newblock {\em The VLDB Journal}, 2022.

\bibitem{fortunato-10}
Santo Fortunato.
\newblock Community detection in graphs.
\newblock {\em Physics Reports}, 486:75--174, 2010.

\bibitem{goldman-71}
Alan~J. Goldman.
\newblock Optimal center location in simple networks.
\newblock {\em Transportation Science}, 5(2):212--221, 1971.

\bibitem{gromov-81}
Mikhael Gromov.
\newblock Groups of polynomial growth and expanding maps.
\newblock {\em Publications Math\'ematiques de l'IH\'ES}, 53:53--78, 1981.

\bibitem{hakimi-64}
Seifollah~Louis Hakimi.
\newblock Optimum locations of switching centers and the absolute centers and
  medians of a graph.
\newblock {\em Operations Research}, 12(3):450--459, 1964.

\bibitem{kannan-04}
Ravi Kannan, Santosh Vempala, and Andrian Vetta.
\newblock On clusterings: Good, bad and spectral.
\newblock {\em Journal of the ACM}, 51(3):497--515, 2004.

\bibitem{khou-17}
Bakh Khoussainov and Toru Takisaka.
\newblock Large scale geometries of infinite strings.
\newblock In {\em LICS}, pages 1--12, 2017.

\bibitem{kroen-08}
Bernhard Kr{\"o}n and R{\"o}gnvaldur~G. M{\"o}ller.
\newblock Quasi-isometries between graphs and trees.
\newblock {\em Journal of Combinatorial Theory, Series B}, 98(5):994--1013,
  2008.

\bibitem{liu-18}
Yike Liu, Tara Safavi, Abhilash Dighe, and Danai Koutra.
\newblock Graph summarization methods and applications: a survey.
\newblock {\em ACM Computing Surveys}, 51(3):62, 2018.

\bibitem{miller-13}
Mirka Miller and \v{S}ir{\'a}n.
\newblock Moore graphs and beyond: {A} survey of the degree/diameter problem.
\newblock {\em The Electronic Journal of Combinatorics}, pages DS14--May, 2013.

\bibitem{metric-embed-ostrovskii-13}
Mikhail~I. Ostrovskii.
\newblock {\em Metric Embeddings}.
\newblock De Gruyter, 2013.

\bibitem{peleg-89}
David Peleg and Alejandro~A. Sch{\"a}ffer.
\newblock Graph spanners.
\newblock {\em Journal of Graph Theory}, 13(1):99--116, 1989.

\bibitem{semple-03}
Charles Semple and Mike~A. Steel.
\newblock {\em Phylogenetics}.
\newblock Oxford University Press, 2003.

\bibitem{slater-82}
Peter~J. Slater.
\newblock Locating central paths in a graph.
\newblock {\em Transportation Science}, 16:1--18, 1982.

\bibitem{thorup-05}
Mikkel Thorup and Uri Zwick.
\newblock Approximate distance oracles.
\newblock {\em Journal of the ACM}, 52(1):1--24, 2005.

\bibitem{wu-chao-04}
Bang-Ye Wu and Kun-Mao Chao.
\newblock {\em Spanning Trees and Optimization Problems}.
\newblock Chapman and Hall/CRC, 2004.

\end{thebibliography}

\end{document}